\documentclass[psamsfonts]{amsart}
\usepackage[backend=biber,
sorting=none]{biblatex}  
\usepackage[utf8]{inputenc}
\usepackage{amsfonts}
\usepackage{hyperref}
\hypersetup{
pdftitle={$SO(3)$ Knot States and the Volume Conjecture},
pdfsubject={Quantum Topology, Geometric Topology, Geometric Analysis},
pdfauthor={Honghuai Fang},
pdfkeywords={Skein theory, Geometric Quantization, Volume Conjecture}
}
\usepackage[justification=centering]{caption}
\usepackage{amsmath}
\usepackage{xcolor}
\usepackage{amsthm}
\usepackage{pdflscape}
\usepackage{pgfplots}
\usepackage{mathrsfs}
\usepackage{amssymb}
\usepackage{float}

\newtheorem{thm}{Theorem}[section]

\newtheorem{lem}[thm]{Lemma}
\newtheorem{conj}[thm]{Conjecture}

\theoremstyle{definition}

\theoremstyle{remark}

\makeatletter
\let\c@equation\c@thm
\makeatother
\numberwithin{equation}{section}

\title{SO(3)-Knot States and the Volume Conjecture}

\author{HONGHUAI FANG}

\begin{document}

\begin{abstract}

We study the alternating subspace of holomorphic sections of a special prequantum line bundle over $SU(2)$-character variety of torus, and show that it is isomorphic to the  projective representation of mapping class group of peripheral torus given by the SO(3) Witten-Chern-Simons theory. We conjecture that the large $r$ asymptotics of $L^2$-norm of $SO(3)$-knot states via geometric quantization capture the simplicial volume of knot complements.
\\\\

\end{abstract}

\maketitle

\section{Introduction}
The Witten-Chern-Simons theory\cite{witten89} is a (2+1)-dimensional topological quantum field theory (TQFT) which associate, for a fixed compact Lie group $G$ and an integer $r$, to each closed oriented surface $\Sigma$ a projective 
representation $V_{G,r}(\Sigma)$ of mapping class group of $\Sigma$ and to each three manifold $M$ with boundary $\Sigma$ a state $Z_{G,r}(M)\in V_{G,r}(\Sigma)$. The first mathematical construction of $SU(2)$-TQFT was given by Reshitikhin and Turaev\cite{reshetikhin1991invariants} via representation of quantum group. Not long after, Lickorish\cite{lickorish1993skein} gave a construction using Kauffman bracket evaluated at $4r$-th roots of unity. For $2(2r+1)$-th roots of unity, there was also a (2+1)-TQFT constructed by BHMV\cite{BHMV95} using skein theory, which is called the $SO(3)$-TQFT.

On the geometric side, Witten implies that $V_{G,r}(\Sigma)$ should be viewed as a subspace of holomorphic section of a prequantum line bundle over character variety of $\Sigma$. Therefore, one can study properties of quantum invariants using geometric approaches. For example, Charles\cite{charles2015knot} study the AJ conjecture and the Witten asymptotic conjecture via geometric quantization of $SU(2)$-TQFT. However, the geometric road to the volume conjecture is beset with difficulties since the $L^2$-norm of $Z_{SU(2),r}(M)$ is of polynomial growth in $r$.

In order to study the volume conjecture from geometric perspective we need to study geometric quantization of $SO(3)$-TQFT. We show that in the case of torus the corresponding geometric space is the alternating subspace of holomorphic sections of $L^r\otimes L^{\frac{1}{2}}\otimes\sigma$ over $Hom(\pi_1(T^2),SU(2))/conj$, denoted by $\mathcal{H}^{alt}_{r+\frac{1}{2}}(j, \delta)$. Our main result is stated as follow.
\begin{thm}Let $V_r'(T^2)$ be the $r$-th $SO(3)$-quantum space of torus defined from skein theory. Then there is an isomorphism $I_r':V_r'(T^2)\rightarrow\mathcal{H}^{alt}_{r+\frac{1}{2}}(j, \delta)$ intertwine with curve operators.
\end{thm}
Therefore, the $SO(3)$-knot state $Z'_r(S^3\backslash K)\in V_r'(T^2)$ defined from skein theory can be view as a holomorphic section in $\mathcal{H}^{alt}_{r+\frac{1}{2}}(j, \delta)$. By analyzing $L^2$-norm of $Z'_r(S^3\backslash K)$ we state the following conjecture, which is somewhat equivalent to the volume conjecture.

\begin{conj}
    For any knot $K$ in $S^3$, let $\operatorname{vol}\left(S^{3} \backslash K\right)$ be the simplicial volume of its complement. Then we have
    \begin{equation}
        \lim _{r \rightarrow+\infty} \frac{2 \pi}{r} \log ||Z'_r(S^3\backslash K)||_{r+\frac{1}{2}}=\operatorname{vol}\left(S^{3} \backslash K\right).
    \end{equation}
    
\end{conj}

Through geometrical methods we can not only use the plentiful tools from geometric analysis but also introduce the physical concept of quantum complexity , inspired by the holographic complexity=volume conjecture\cite{brown2016complexity}. Moreover, complexity emerges naturally through Mahler measure of A-polynomials\cite{boyd2002mahler} and topological complexityof 3-manifolds\cite{francaviglia2012stable}.  We expect that the geometric picture of $SO(3)$-TQFT brings to light some relationships between math, physics and quantum computation.

\section{SO(3)-knot states via skein theory}
In this section we review the $SO(3)$-TQFT following the skein theoretical approach by BHMV.  The \textit{Kauffman bracket skein module} $K_A(M)$ of oriented 3-manifold $M$ is defined as the quotient of the free $\mathbb{Z}[A,A^{-1}]$-module generated by the isotopic classes of framed links in $M$ by Kauffman relations. For example, there is a canonical isomorphism between $K_A(S^3)$ and $\mathbb{Z}[A,A^{-1}]$ via Kauffman bracket. The Kauffman bracket skein module of solid torus $K_A(D^2\times S^1)$ can be viewed as the module $\mathbb{Z}[A,A^{-1}][z]$. Fix a framed link $L$ with $m$ components in $S^3$, one can define a $\mathbb{Z}[A,A^{-1}]$-multilinear map 
$$\langle\ \ \ , \ldots ,\ \ \  \rangle_L : K_A(D^2\times S^1)^{\otimes m} \rightarrow \mathbb{Z}[A,A^{-1}],$$
as follow. For monomials $z^{i_k}\in \mathbb{Z}[A,A^{-1}][z]\cong K_A(D^2\times S^1),$ $\langle z^{i_1} , \ldots ,z^{i_m}  \rangle_L$ is simply the Kauffman bracket of the framed link obtained from $L$ by replacing the $k$-th component by $i_k$ parallel copies, and extend $\mathbb{Z}[A,A^{-1}]$-linearly on the whole $K_A(D^2\times S^1).$ Let $e_n$ be the $n$-th Chebyshev polynomial in $z$, then the \textit{colored Jones polynomials} of an oriented  knot $K$ in $S^3$ are defined by $$J_{K,n+1}(t)=\frac{\big((-1)^n A^{n^2+2n}\big)^{w(K)} \langle e_n \rangle_K}{(-1)^{n}[n+1]}\big|_{t=A^4},$$ where $w(K)$ is the writhe number of $K$ and $[n]=\frac{A^{2n}-A^{-2n}}{A^{2}-A^{-2}}$ is the quantum number.

From now on, let $A$ be a primitive $(4r+2)$-th root of unity for an integer $r\geqslant 3$. Consider some special elements of $K_A(D^2\times S^1)$, called the \textit{$SO(3)$-Kirby coloring}, defined by 
$$\omega_{r}'=\underset{i=0}{\overset{r-1}{\sum}}\langle e_{i} \rangle e_{i}$$ for any integer $r$. We also for any $r$ introduce
$$\eta_{r}'=\frac{2\sin(\frac{2\pi}{2r+1})}{\sqrt{2r+1}},\\\ \kappa'_{r}=\eta_r' \langle \omega_r' \rangle_{U_+}.$$

The \textit{$SO(3)$-Reshetikhin-Turaev invariants} of closed oriented 3-manifold $M$ obtained from $S^3$ by doing surgery along a framed link $L$ with number of components $m$ and signature $\sigma$ are defined by
$$\langle M\rangle_{r}'=(\eta_{r}')^{1+m}\ (\kappa_{r}')^{-\sigma}\ \langle \omega_{r}', \dots, \omega_{r}'\rangle _{L}.$$
Let $\Tilde{L}$ be a framed link in $M.$ Then, the SO(3)-Reshetikhin-Turaev invariants of the pair $(M,\Tilde{L})$ are defined by
$$\langle M, \Tilde{L} \rangle_{r}'=(\eta_{r}')^{1+m}\ (\kappa_{r}')^{-\sigma}\ \langle \omega_{r}', \dots, \omega_{r}', 1 \rangle _{L\cup \Tilde{L}}.$$

Using the universal construction, the $SO(3)$-Reshetikhin-Turaev invariants determine a family $(V_r',Z_r')$ of quantization functors, called the $SO(3)$-TQFT, from cobordism categories to the category of complex vector spaces.  This TQFT associate to any closed oriented surface $\Sigma$ a hermitian complex vector space $V_r'(\Sigma)$ and to the cobordism $M$ with $\partial M=\Sigma$ and $L$ in $M$ a vector $Z_r'(M,L)\in V_r'(\Sigma)$. In particular, $V_r'(T^2)$ is identified with quotients of the Kauffman bracket skein module $K_A(D^2\times S^1)$, and  the vectors $\{ e_0,\dots,e_{r-1}\}$ form a Hermitian basis of $Z_r'(T^2)$. For a knot complement $S^3\backslash K$ one can show that $$Z_r'(S^3\backslash K)=\underset{0 \leqslant n \leqslant r-1}{\sum}\eta_{r}'\langle e_{n}\rangle_K e_n,$$ which is called the \textit{$SO(3)$-knot states} of $K$.

Fix an oriented diffeomorphism $\psi:\Sigma=\partial(S^3\backslash K)\rightarrow S^1\times S^1$. For any curve $\gamma$ in $\Sigma$, we define the endomorphism $Z_r'(\gamma)$ of $V_r'(\Sigma)$ by 
$$Z_r'(\gamma)=Z_r'(\Sigma\times [0,1], \gamma\times [\frac{1}{3},\frac{2}{3}]).$$
In particular, for $\mu=\psi^{-1}(S^1\times\{1\})$ and $\lambda=\psi^{-1}(\{1\}\times S^1)$ we have two operators $Z_r'(\mu)$ and $Z_r'(\lambda)$. For any $g\in SL(2,\mathbb{Z})$, we define the endomorphism $Z_r'(g)$ of $V_r'(\Sigma)$ by 
$$Z_r'(g)=Z_r'(M_g),$$ where $M_g=S^1\times S^1\times [0,1]$ is the mapping cylinder whose boundary
is identified with $-\Sigma\sqcup\Sigma$ through $\psi\sqcup\psi\circ g$. Consider the generators of  $SL(2,\mathbb{Z})$ 
$$T=\begin{pmatrix}
1 & 1\\
0 & 1
\end{pmatrix}, S=\begin{pmatrix}
0 & -1\\
1 & 0
\end{pmatrix}.$$

It is showed in BHMV that for $n=0,1,\dots,r-1$ we have
\begin{equation}\label{mulambda}
 Z_r'(\mu)e_n=-2cos(\frac{2(n+1)\pi}{2r+1})e_n,\\\ Z_r'(\lambda)e_n=-(e_{n-1}+e_{n+1}),   
\end{equation}
and
\begin{equation}\label{TS}
Z_r'(T)e_n=exp(\frac{\pi i(n^2+2n)}{2r+1})e_n,\\\ Z_r'(S)e_n=\underset{m \in \mathbb{Z} /(2r+1) \mathbb{Z}}{\sum}\frac{2ie^{-\frac{\pi i}{4}}}{\sqrt{2r+1} } \sin(\frac{2\pi(m+1)(n+1)}{2r+1})e_m.
\end{equation}

\section{SO(3)-knot states via geometric quantization}
In this section we construct a special family of prequantum line bundles over the character variety of torus, or the moduli space of flat connections on torus. Then we will show that the space of holomorphic sections of such line bundle is unitary isomorphic to the corresponding $SO(3)$-knot state space.

Let $(V,\omega)$ be a real 2-dimensional symplectic vector space equipped with a compatible linear complex structure $j$. Let $\alpha\in \Omega^1(V,\mathbb{C})$ be given by $\alpha_x(y)=\frac{1}{2}\omega(x,y)$ and endow the trivial line bundle $L=V\times \mathbb{C}$ with connection $\nabla=d-i\alpha$, the standard hermitian structure. and the unique holomorphic structure $h$, making it a prequantum line bundle. Given a lattice $\Lambda\subset V$, then there is a prequantum line bundle over $T_{\Lambda}\triangleq V/\Lambda$ if and only if the symplectic volume of the fundamental domain $D$ is an integer multiple of $2\pi$.

We consider two types of half line bundle. For the sake of simplicity let the volume of $D$ equal to $4\pi$. The induced simplectic form of $T_{\Lambda}$ is in $H^2(T_{\Lambda},4\pi\mathbb{Z})$, and therefore the induced prequantum line bundle over $T_{\Lambda}$ has a natural squart root since the first Chern class $c_1(L)=[\frac{\omega}{2\pi}]$ is even. We denote this half form by $L^{\frac{1}{2}}$ and endow it with the standard hermitian structure $h_{\frac{1}{2}}$. Meanwhile for the canonical line bundle $K_j=\left\{\alpha \in \Omega^1(V,\mathbb{C})|\alpha(j \cdot)=i \alpha\right\}$ over $V$ we choose a half form $\delta$ of $K_j$ with an isomorphism $\varphi: \delta^{\otimes 2} \rightarrow K_j$. $K_j$ has a natural scalar product such that the square of the norm of $\alpha$ is $i \alpha \wedge \frac{\bar{\alpha} }{\omega}$. We endow $\delta$ with the hermitian structure $h_{\delta}$ making $\varphi$ an isometry. 

Choose a basis $\{\mu,\lambda\}$ of $\Lambda$ such that $\omega(\mu,\lambda)=4\pi$. Let $\tau=a+bi$ be a complex number such that $\lambda=a\mu+bj\mu$. Let $p,q:V\rightarrow\mathbb{R}$ be the linear coordinate dual to $\mu,\lambda$. Thus $z=p+\tau q$ is a holomorphic coordinate of $(V,j)$. In this coordinate $\omega=4\pi dp\wedge dq$. A section $t\in C^{\infty}(V,L)$ is holomorphic if and only if
$$
0=\nabla_{\bar{Z}} t=\frac{\partial t}{\partial p}-\frac{1}{\tau} \frac{\partial t}{\partial q}+2 \pi i\left(q+\frac{p}{\tau}\right) t
$$
for $Z=\mu-\frac{\tau}{|\tau|}\lambda$. One readily checks that $t(p,q)=exp(2\pi i q(p+\tau q))$ is a homlomorphic section of $L$. Similarly $t'(p,q)=exp(\pi i q(p+\tau q))$ is a holomorphic section of $L^{\frac{1}{2}}$. Therefore any holomorphic section of $L^{r}\otimes L^{\frac{1}{2}}$ is of the form $gt^{r+\frac{1}{2}}$ where $t^{r+\frac{1}{2}}(p,q)\triangleq t^{r}(p,q)t'(p,q)$ and $g:V\rightarrow \mathbb{C}$ is holomorphic, which means that it satisfies
\begin{equation}\label{holog}
    \frac{\partial g}{ \partial q}=\tau \frac{\partial g }{\partial p}.
\end{equation}

Consider $V\times U(1)$ as a Heisenberg group with the product 
$$
(x, u) \cdot(y, v)=\left(x+y, u v \exp \left(\frac{ (2r+1)}{4} i\omega(x, y)\right)\right),
$$
and denote it by $G_{r+\frac{1}{2}}$. The same formula defines an action of $G_{r+\frac{1}{2}}$ on $L^{r}\otimes L^{\frac{1}{2}}$ preserving the connection and the hermitian structure. Consider the subgroup $\Lambda\times\{1\}$ of $G_{r+\frac{1}{2}}$. For any $x\in \Lambda$ denote by $T^{*}_x$ the pullback by the action of $(x,1)\in\Lambda\times\{1\}$, which means that for any $\Phi\in\Gamma(V,L^{r}\otimes L^{\frac{1}{2}}),$
\begin{equation}\label{T}
   (T^{*}_x \Phi)(y)=exp\left(\frac{ (2r+1)}{4} i\omega(x, y)\right)\Phi(x+y).
\end{equation}

One can compute that
\begin{equation}\label{TMN}
 (T^{*}_{m\mu+n\lambda} t^{r+\frac{1}{2}})(p,q)=exp\left((2r+1)\pi i\left(\tau n^2+2n(p+\tau q)\right)\right)t^{r+\frac{1}{2}}(p,q).
\end{equation}
Therefore if $gt^{r+\frac{1}{2}}\in H^0(V,L^{r}\otimes L^{\frac{1}{2}})$ is invariant under the action of $\Lambda\times\{1\}$, then 
\begin{equation}\label{gmn}
 g(p+m,q+n)=exp\left(-(2r+1)\pi i\left(\tau n^2+2n(p+\tau q)\right)\right)g(p,q). 
\end{equation}
In particular, $g(p+1, q)=g(p, q)$. Consequently, 
$$
g(p, q)=\sum_{m \in \mathbb{Z}} g_m(q) \exp (2 m\pi i  p)
$$
for some smooth functions $g_m: \mathbb{R} \rightarrow \mathbb{C}$. The condition (\hyperref[holog]{3.1}) implies
$$
g_m'=2m\pi i\tau g_m
$$
for every $m\in\mathbb{Z}$, so there is a sequence $(\rho_m)_{m\in\mathbb{Z}}$ of complex number such that $g_m(q)=\rho_m exp(2m\pi i\tau q).$ We obtain that 
\begin{equation}\label{g}
    g(z)=\sum_{m \in \mathbb{Z}}\rho_m exp(2m\pi i z).
\end{equation}.

On the one hand, from (\hyperref[g]{3.5}) we have that 
$$
g(z+n\tau)=\sum_{m \in \mathbb{Z}}\rho_m exp(2mn\pi i \tau)exp(2m\pi i z).
$$
On the other hand take $m=0$ in (\hyperref[gmn]{3.4}) we have that
$$
 g(z+n\tau)=exp\left(-(2r+1)\pi i\left(\tau n^2+2nz\right)\right)g(z). 
$$
Consequently we obtain that 
\begin{equation}\label{rho}
  \rho_{m+\left(2r+1\right)n}=exp\left(4n\pi i\tau\left(2m+\left(2r+1\right)n\right)\right)\rho_m,  
\end{equation} which means that the sequence $(\rho_m)_{m\in\mathbb{Z}}$ is determined by $\rho_0, \ldots, \rho_{2r}$. Note that any choice of such coeffcients yields an element of $\Lambda\times\{1\}$-invariant subspace of  $H^0(V,L^{r}\otimes L^{\frac{1}{2}})$. 

Let $G_{r+\frac{1}{2}}$ act trivially on $\delta$. Denote by $\mathcal{H}_{r+\frac{1}{2}}(j,\delta)$ the $\Lambda\times\{1\}$-subspace 
of $H^0(V,L^r\otimes L^{\frac{1}{2}}\otimes\delta)$. It has to be considered as the space of
holomorphic sections of the line bundle $L^r\otimes L^{\frac{1}{2}}\otimes\delta/(\Lambda\times\{1\})$ over the torus $T^2_{\lambda}=V/\Lambda$. We have just shown that $\mathcal{H}_{r+\frac{1}{2}}(j,\delta)$ has dimension $2r+1$. Define the inner product of $\Psi_1, \Psi_2 \in \mathcal{H}_{r+\frac{1}{2}}(j, \delta)$ by
$$
\left\langle\Psi_1, \Psi_2\right\rangle'_{r+\frac{1}{2}}=\int_D \left(h^{\otimes r}\otimes h_{\frac{1}{2}}\otimes h_{\delta}\right)_x\left(\Psi_1(x), \Psi_2(x)\right)\omega.
$$

One readily checks that $\mathcal{H}_{r+\frac{1}{2}}(j, \delta)$ is a representation of $\frac{\Lambda}{2r+1}\times\{1\}$, which gives us a canonical orthonormal basis of $\mathcal{H}_{r+\frac{1}{2}}(j, \delta)$ as follow.

\begin{lem}\label{commute}
For integer $r\geqslant 3 $, one has\\
(i) $T^{*}_{\frac{\mu}{2r+1}} t^{r+\frac{1}{2}}=t^{r+\frac{1}{2}}$\\
(ii) $T^{*}_{\frac{\mu}{2r+1}}T^{*}_{\frac{\lambda}{2r+1}}=exp(\frac{2\pi i}{2r+1})T^{*}_{\frac{\lambda}{2r+1}}T^{*}_{\frac{\mu}{2r+1}}.$
\end{lem}

\begin{proof}
Taking $m=\frac{1}{2r+1},n=0$ in (\hyperref[TMN]{3.3}) we obtain (i). From (\hyperref[T]{3.2}) we show that
$$(T^{*}_{\frac{\mu}{2r+1}}\Psi)(p,q)=exp(-\pi i q)\Psi(p+\frac{1}{2r+1},q),(T^{*}_{\frac{\lambda}{2r+1}}\Psi)(p,q)=exp(\pi i p)\Psi(p,q+\frac{1}{2r+1}),$$
and we deduce (ii) from these formulas.
\end{proof}

\begin{thm}\label{onb}
There exists an orthonormal
basis $\left(\Psi_{l}\right)_{l \in \mathbb{Z} /(2r+1) \mathbb{Z}}$ of $\mathcal{H}_{r+\frac{1}{2}}(j, \delta)$ such that 
$$T^{*}_{\frac{\mu}{2r+1}}\Psi_l=exp(\frac{2l\pi i}{2r+1})\Psi_l, T^{*}_{\frac{\lambda}{2r+1}}\Psi_l=\Psi_{l+1}.$$
\end{thm}
\begin{proof}
Firstly, from the proof of Lemma \hyperref[commute]{3.7} we observe that $T^{*}_{\frac{\mu}{2r+1}}$ and $T^{*}_{\frac{\lambda}{2r+1}}$ are unitary. Secondly, the  commutative property implies that if $\lambda_0$ is an eigenvalue of $T^{*}_{\frac{\mu}{2r+1}}$ and $\Psi_0$ is an unit eigenvector associated with $\lambda_0$, then $T^{*}_{\frac{\lambda}{2r+1}}\Psi_0$ is an eigenvector of $T^{*}_{\frac{\mu}{2r+1}}$ with eigenvalue $exp(\frac{2\pi i}{2r+1})\lambda_0$. Consequently, $T^{*}_{\frac{\mu}{2r+1}}$ has 2r+1 distinct eigenvalues 
$$\lambda_l=exp(\frac{2l\pi i}{2r+1})\lambda_0, l=0,\dots,2r,$$
and $(\Psi_l=(T^{*}_{\frac{\lambda}{2r+1}})^l\Psi_0)_{l \in \mathbb{Z} /(2r+1) \mathbb{Z}}$ forms an orthonormal basis of $\mathcal{H}_{r+\frac{1}{2}}(j, \delta)$ since the dimension of this quantum space is 2r+1.

Let $\Omega_{\mu}$ be a vector of $\delta$ such that $\varphi(\Omega_{\mu}^2)(\mu)=1$. Define
$$
\Psi_0=(\frac{4\pi}{2r+1})^{\frac{1}{4}}g_0t^{r+\frac{1}{2}}\Omega_{\mu},
$$
where $g_0(z)=\sum_{m\in\mathbb{Z}}exp\left(m\pi i\left((4r+2)z+(2r+1)m\tau\right)\right)$. Observed that by taking $\rho_0=1,\rho_1=\dots=\rho_{2r}=0$ in (\hyperref[rho]{3.6}) we obtain $g_0$. Therefore, $\Psi_0$ is a element of $\mathcal{H}_{r+\frac{1}{2}}(j, \delta)$. Next, since $T^{*}_{\frac{\mu}{2r+1}} t^{r+\frac{1}{2}}=t^{r+\frac{1}{2}}$ and $g_0(z+\frac{1}{2r+1})=g_0(z)$ we have that $T^{*}_{\frac{\mu}{2r+1}}\Psi_0=\Psi_0$. Moreover, from a standard computation we obtain that $$||\Omega_{\mu}||_{r+\frac{1}{2}}=(\frac{b}{2\pi})^{\frac{1}{4}}, ||g_0t^{r+\frac{1}{2}}||_{r+\frac{1}{2}}=(\frac{8\pi^2}{(2r+1)b})^{\frac{1}{4}}.$$ We conclude that $\Psi_0$ is an unit eigenvector of $T^{*}_{\frac{\mu}{2r+1}}$ with eigenvalue 1.
\end{proof}

The $SU(2)$-character variety of torus, or the moduli space of flat $SU(2)$-connections on  peripheral torus, is the variety $Hom(\pi_1(T^2),SU(2))/conj$, which can be view as $V/(\Lambda\rtimes \mathbb{Z}_2)$. Therefore, we focus on the space
$$
\mathcal{H}^{alt}_{r+\frac{1}{2}}(j, \delta)\triangleq\{\Psi\in\mathcal{H}_{r+\frac{1}{2}}(j, \delta)|\Psi(x)=-\Psi(-x)\}.
$$
Theorem \hyperref[onb]{3.8} implies that 
$$\left(\Phi_l=\frac{1}{\sqrt{2}}\left(\Psi_l-\Psi_{-l}\right)\right)_{l=1,\dots,r}
$$
forms an orthonormal basis of $\mathcal{H}^{alt}_{r+\frac{1}{2}}(j, \delta)$.

Fix a knot $K$ in $S^3$ and let $\Sigma$ be the boundary of $S^3\backslash K$. choosing an oriented diffeomorphism $\psi:\Sigma\rightarrow S^1\times S^1$ and let $\mu$ and $\lambda$ be the homology classes of $\psi^{-1}\left(S^1 \times\right.$ $\{1\})$ and $\psi^{-1}\left(\{1\} \times S^1\right)$ respectively. On the topological side, we have a hermitian space $V_r'(\Sigma)$  with an orthonormal basis $(e_n)_{n=0,\dots,r-1}$ . On the geometric side, the variety $Hom(\pi_1(\Sigma),SU(2))/conj$ can be viewed as $H_1(\Sigma, \mathbb{R})/(H_1(\Sigma, \mathbb{Z})\rtimes\mathbb{Z}_2)$, where we endow $H_1(\Sigma, \mathbb{R})$ with a symplectic form $\omega(x,y)=4\pi x\cdot y$, a linear complex sturcture $j$
and a metaplectic form $\delta$. Choose $\{\mu,\lambda\}$ as a basis of $H_1(\Sigma, \mathbb{Z})$ we have a hermitian space $\mathcal{H}^{alt}_{r+\frac{1}{2}}(j, \delta)$ with an orthonormal basis $\left(\Phi_l\right)_{l=1,\dots,r}.$ For any oriented curve $\gamma\subset\Sigma$ we define an endomorphim of $\mathcal{H}^{alt}_{r+\frac{1}{2}}(j, \delta)$ by
$$T_{r+\frac{1}{2}}(\gamma)=-\left(T_{\frac{\gamma}{2r+1}}^*+T_{-\frac{\gamma}{2r+1}}^*\right).$$
Now we show that these two spaces are isomorphic naturally with respect to the projective action of the mapping class group of $\Sigma$.

\begin{thm}
For any oriented curve $\gamma\subset\Sigma$, there is an isomorphism $I_r':V_r'(\Sigma)\rightarrow\mathcal{H}^{alt}_{r+\frac{1}{2}}(j, \delta)$ such that the following diagram commutes:
\begin{center}

\tikzset{every picture/.style={line width=0.75pt}} 

\begin{tikzpicture}[x=0.75pt,y=0.75pt,yscale=-1,xscale=1]

\draw    (509,42.96) -- (587.8,41.99) ;
\draw [shift={(589.8,41.96)}, rotate = 179.29] [color={rgb, 255:red, 0; green, 0; blue, 0 }  ][line width=0.75]    (10.93,-3.29) .. controls (6.95,-1.4) and (3.31,-0.3) .. (0,0) .. controls (3.31,0.3) and (6.95,1.4) .. (10.93,3.29)   ;
\draw    (482,58.96) -- (481.81,106.96) ;
\draw [shift={(481.8,108.96)}, rotate = 270.23] [color={rgb, 255:red, 0; green, 0; blue, 0 }  ][line width=0.75]    (10.93,-3.29) .. controls (6.95,-1.4) and (3.31,-0.3) .. (0,0) .. controls (3.31,0.3) and (6.95,1.4) .. (10.93,3.29)   ;
\draw    (636,56.96) -- (636.77,106.96) ;
\draw [shift={(636.8,108.96)}, rotate = 269.12] [color={rgb, 255:red, 0; green, 0; blue, 0 }  ][line width=0.75]    (10.93,-3.29) .. controls (6.95,-1.4) and (3.31,-0.3) .. (0,0) .. controls (3.31,0.3) and (6.95,1.4) .. (10.93,3.29)   ;
\draw    (515,118.96) -- (593.8,117.99) ;
\draw [shift={(595.8,117.96)}, rotate = 179.29] [color={rgb, 255:red, 0; green, 0; blue, 0 }  ][line width=0.75]    (10.93,-3.29) .. controls (6.95,-1.4) and (3.31,-0.3) .. (0,0) .. controls (3.31,0.3) and (6.95,1.4) .. (10.93,3.29)   ;

\draw (457,33.36) node [anchor=north west][inner sep=0.75pt]    {$V_{r} '( \Sigma )$};
\draw (595,29.36) node [anchor=north west][inner sep=0.75pt]    {$\mathcal{H}_{r+\frac{1}{2}}^{alt} (j,\delta )$};
\draw (537,20.36) node [anchor=north west][inner sep=0.75pt]    {$I_{r} '$};
\draw (539,95.36) node [anchor=north west][inner sep=0.75pt]    {$I_{r} '$};
\draw (599,106.36) node [anchor=north west][inner sep=0.75pt]    {$\mathcal{H}_{r+\frac{1}{2}}^{alt} (j,\delta )$};
\draw (458,109.36) node [anchor=north west][inner sep=0.75pt]    {$V_{r} '( \Sigma )$};
\draw (435,70.36) node [anchor=north west][inner sep=0.75pt]    {$Z_{r} '( \gamma )$};
\draw (643,70.36) node [anchor=north west][inner sep=0.75pt]    {$T_{r+\frac{1}{2}}( \gamma )$};

\end{tikzpicture}    
\end{center}
\end{thm}

\begin{proof}
 We have shown that both $V_r'(\Sigma)$ and $\mathcal{H}^{alt}_{r+\frac{1}{2}}(j, \delta)$ have dimension $r$ and that given an oriented diffeomorphism $\psi:\Sigma\rightarrow S^1\times S^1$ we have orthonormal basis $(e_l)_{l=0,\dots,r-1}$ and $(\Phi_l)_{l=1,\dots,r}$ respectively. Define an isomorphism $I_r'$ from $V_r'(\Sigma)$ to $\mathcal{H}^{alt}_{r+\frac{1}{2}}(j, \delta)$ by
 $$
 I_r'(e_l)=\Psi_{l+1}, l=0,\dots,r-1.
 $$

 Let $\tilde{\psi}:\Sigma\rightarrow S^1\times S^1$ be another oriented diffeomorphism and let $(\tilde{e_l})_{l=0,\dots,r-1}$ and $(\tilde{\Phi}_l)_{l=1,\dots,r}$ be the corresponding orthonormal basis respectively.. Define another isomorphism $\tilde{I_r'}$ by  $\tilde{I_r'}(\tilde{e_l})=\tilde{\Psi}_{l+1}$ for $l=0,\dots,r-1$. 
 
 We show that $I_r'=\tilde{I_r'}$ as follow. The automorphism 
 $\tilde{\psi}\circ\psi^{-1}:S^1\times S^1\rightarrow S^1\times S^1$ is isotopic to a linear map in $SL(2,\mathbb{Z})$. Thus we may assume that $\tilde{\psi}\circ\psi^{-1}$ is $T$ or $S$. On the one hand, by (\hyperref[TS]{2.2}) we have that 
 \begin{equation}
     \tilde{e_l}=Z_r'(\tilde{\psi}\circ\psi^{-1})e_l=\begin{cases}

      exp(\frac{\pi i(l^2+2l)}{2r+1})e_l,& \tilde{\psi}\circ\psi^{-1}=T\\
      \underset{m \in \mathbb{Z} /(2r+1) \mathbb{Z}}{\sum}\frac{2ie^{-\frac{\pi i}{4}}}{\sqrt{2r+1} } \sin(\frac{2\pi(m+1)(l+1)}{2r+1})e_m,&\tilde{\psi}\circ\psi^{-1}=S.
      
     \end{cases}
 \end{equation}

 On the other hand, let $(\mu,\lambda)$ and $(\tilde{\mu},\tilde{\lambda})$ be two basis of $H_1(\Sigma,\mathbb{Z})$ defined from $\psi$ and $\tilde{\psi}$ respectively. Then we have that
 \begin{equation}
     (\tilde{\mu},\tilde{\lambda})=\begin{cases}
         (\mu+\lambda,\lambda),&\tilde{\psi}\circ\psi^{-1}=T\\
         (\lambda,-\mu),&\tilde{\psi}\circ\psi^{-1}=S.
     \end{cases}
 \end{equation}

We show that the relationship between $\tilde{\Phi}_l$ and $\Phi_l$ is the same as the relationship between $\tilde{e_l}$ and $e_l$. In the first case, we have that $\tilde{p}=p+q$, $\tilde{q}=q$, $\tilde{\tau}=\frac{\tau}{\tau +1}$ and $\Omega_{\tilde{\mu}}=\Omega_{\mu+\lambda}=(\frac{1}{1+\tau})^{\frac{1}{2}}\Omega_{\mu}$. Then we compute that
\begin{equation}
\begin{split}
    t^{r+\frac{1}{2}}_{\tilde{\mu},\tilde{\lambda}}&=exp((2r+1)\pi i\tilde{q}(\tilde{p}+\tilde{\tau}\tilde{q}))\\
&=exp((2r+1)\pi i q (p+q+\tilde{\tau}q))\\
&=exp((2r+1)\pi i q^2 (1+\tilde{\tau}-\tau))t^{r+\frac{1}{2}}_{\mu,\lambda},
\end{split}
\end{equation}

\begin{equation}
    \begin{split}
       (T^{*}_{\frac{\tilde{\lambda}}{2r+1}})^l(g_{0;\tilde{\mu},\tilde{\lambda}})(0)&=\underset{m\in\mathbb{Z}}{\sum}exp((2r+1)\pi i \tilde{\tau} m^2+2\pi i \tilde{\tau} l m) \\
       &=exp(\frac{-\pi i \tilde{\tau} l^2}{2r+1})\underset{m\in\mathbb{Z}}{\sum}exp((2r+1)\pi i \tilde{\tau}(m+\frac{l}{2r+1})^2)\\
       &=e^{\frac{-\pi i \tilde{\tau} l^2}{2r+1}}(\frac{i}{(2r+1)\tilde{\tau}})^{\frac{1}{2}}\underset{m\in\mathbb{Z}}{\sum}e^{\frac{-\pi i m^2}{(2r+1)\tilde{\tau}}}e^{\frac{2\pi i lm}{2r+1}}\\
              &=e^{\frac{-\pi i \tilde{\tau} l^2}{2r+1}}(\frac{i}{(2r+1)\tilde{\tau}})^{\frac{1}{2}}\underset{m\in\mathbb{Z}}{\sum}e^{\frac{-\pi i m^2}{(2r+1)\tau}}e^{\frac{-\pi i m^2+2\pi i lm}{2r+1}}\\
       &=e^{\frac{-\pi i (\tilde{\tau}-\tau) l^2}{2r+1}}e^{\frac{-\pi i}{2r+1}}(\frac{1}{1+\tau})^{-\frac{1}{2}}(e^{\frac{-\pi i \tau l^2}{2r+1}}(\frac{i}{(2r+1)\tau})^{\frac{1}{2}}\underset{m\in\mathbb{Z}}{\sum}e^{\frac{-\pi i m^2}{(2r+1)\tau}}e^{\frac{2\pi i lm}{2r+1}})\\
       &=e^{\frac{-\pi i (\tilde{\tau}-\tau) l^2}{2r+1}}e^{\frac{-\pi i}{2r+1}}(\frac{1}{1+\tau})^{-\frac{1}{2}} (T^{*}_{\frac{\lambda}{2r+1}})^l(g_{0;\mu,\lambda})(0).
    \end{split}
\end{equation}

Consequently, we have that
\begin{equation}
    \begin{split}
      (T^{*}_{\frac{\tilde{\lambda}}{2r+1}})^l(\tilde{\Psi}_0)(0) &= (\frac{4\pi}{2r+1})^{\frac{1}{4}}\Omega_{\tilde{\mu}}  (T^{*}_{\frac{\tilde{\lambda}}{2r+1}})^l(t^{r+\frac{1}{2}}_{\tilde{\mu},\tilde{\lambda}})(0)(T^{*}_{\frac{\tilde{\lambda}}{2r+1}})^l(g_{0;\tilde{\mu},\tilde{\lambda}})(0)\\
      &=(\frac{4\pi}{2r+1})^{\frac{1}{4}}(\frac{1}{1+\tau})^{\frac{1}{2}}\Omega_{\mu}e^{\frac{\pi i  (1+\tilde{\tau}-\tau)l^2}{2r+1}} (T^{*}_{\frac{\lambda}{2r+1}})^l(t^{r+\frac{1}{2}}_{\mu,\lambda})(0)\\
      &\times e^{\frac{-\pi i (\tilde{\tau}-\tau) l^2}{2r+1}}e^{\frac{-\pi i}{2r+1}}(\frac{1}{1+\tau})^{-\frac{1}{2}} (T^{*}_{\frac{\lambda}{2r+1}})^l(g_{0;\mu,\lambda})(0)\\
     &=e^{\frac{\pi i  (l^2-1)}{2r+1}}(\frac{4\pi}{2r+1})^{\frac{1}{4}}\Omega_{\mu}  (T^{*}_{\frac{\lambda}{2r+1}})^l(t^{r+\frac{1}{2}}_{\mu,\lambda})(0)(T^{*}_{\frac{\lambda}{2r+1}})^l(g_{0;\mu,\lambda})(0)\\
     &=e^{\frac{\pi i  (l^2-1)}{2r+1}}(T^{*}_{\frac{\lambda}{2r+1}})^l(\Psi_0)(0),
\end{split}
\end{equation}
which implies that $\tilde{\Psi}_l=e^{\frac{\pi i  (l^2-1)}{2r+1}}\Psi_l$. It follows that
$$
\Phi_l=exp(\frac{\pi i  (l^2-1)}{2r+1})\Phi_l, l=1,\dots,r,
$$
which shows the result for the first case.

 In the second case, we have that $\tilde{p}=q$, $\tilde{q}=-p$, $\tilde{\tau}=-\frac{1}{\tau}$ and $\Omega_{\tilde{\mu}}=\Omega_{\lambda}=(\frac{1}{\tau})^{\frac{1}{2}}\Omega_{\mu}$. Then we compute that 

 \begin{equation}
     \begin{split}
         \tilde{\Psi}_0(0)&=(\frac{4\pi}{2r+1})^{\frac{1}{4}}\Omega_{\tilde{\mu}}\underset{m\in\mathbb{Z}}{\sum}exp((2r+1)\pi i \tilde{\tau} m^2)\\
         &=(\frac{4\pi}{2r+1})^{\frac{1}{4}}(\frac{1}{\tau})^{\frac{1}{2}}\Omega_{\mu}\underset{m\in\mathbb{Z}}{\sum}exp(\frac{-(2r+1)\pi i m^2}{\tau})\\
         &=(\frac{4\pi}{2r+1})^{\frac{1}{4}}\Omega_{\mu}(\frac{1}{(2r+1)i})^{\frac{1}{2}}(\frac{(2r+1)i}{\tau})^{\frac{1}{2}}\underset{m\in\mathbb{Z}}{\sum}exp(\frac{-(2r+1)\pi i m^2}{\tau})\\
         &=(\frac{1}{(2r+1)i})^{\frac{1}{2}}(\frac{4\pi}{2r+1})^{\frac{1}{4}}\Omega_{\mu}\underset{m\in\mathbb{Z}}{\sum}exp(\frac{\pi i \tau m^2}{2r+1})\\
         &=e^{-\frac{\pi i}{4}}(2r+1)^{-\frac{1}{2}}\underset{m \in \mathbb{Z} /(2r+1) \mathbb{Z}}{\sum}T^{*}_{\frac{m\lambda}{2r+1}}(\Psi_0)(0).
     \end{split}
 \end{equation}
 It follows that 
 \begin{equation}
     \begin{split}
        \tilde{\Psi}_l&=(T^{*}_{\frac{\tilde{\lambda}}{2r+1}})^l(\tilde{\Psi}_0) \\
        &=e^{-\frac{\pi i}{4}}(2r+1)^{-\frac{1}{2}}\underset{m \in \mathbb{Z} /(2r+1) \mathbb{Z}}{\sum}(T^{*}_{\frac{-l\mu}{2r+1}})^l(\Psi_m)\\
        &=e^{-\frac{\pi i}{4}}(2r+1)^{-\frac{1}{2}}\underset{m \in \mathbb{Z} /(2r+1) \mathbb{Z}}{\sum}e^{\frac{-\pi i l m}{2r+1}}(\Psi_m).        
     \end{split}
 \end{equation}
Thus, we have that 
\begin{equation}
    \tilde{\Phi}_l=2ie^{-\frac{\pi i}{4}}(2r+1)^{-\frac{1}{2}}\underset{m \in \mathbb{Z} /(2r+1) \mathbb{Z}}{\sum}\sin(\frac{\pi l m}{2r+1})\Phi_m,
\end{equation}
which show the result of the second case.

 Therefore, $I_r'$ is independent of choices of $\psi$. For any $\gamma\subset\Sigma$ let $\psi$ be an oriented diffeomorphism from $\Sigma$ to $S^1\times S^1$ such that $\psi(\gamma)=S^1\times\{1\}$. Then we have that
\begin{equation}
    \begin{split}
        I'_r\circ Z'_r(\gamma)(e_{l-1})&=-2\cos(\frac{2\pi l}{2r+1})I_r'(e_{l-1})\\
        &=-(e^{\frac{2\pi i l}{2r+1}}+e^{-\frac{2\pi i l}{2r+1}})\Phi_l\\
        &=-(T^{*}_{\frac{\gamma}{2r+1}}+T^{*}_{\frac{-\gamma}{2r+1}})(\frac{1}{\sqrt{2}}(\Psi_l-\Psi_{-l}))\\
        &=T_{r+\frac{1}{2}}(\gamma)(\Phi_l)=T_{r+\frac{1}{2}}(\gamma)\circ I'_r(e_{l-1}). 
    \end{split}
\end{equation}
 
\end{proof}

Consequently, the $SO(3)$-knot state $Z'_r(S^3\backslash K)\in V_r'(\Sigma)$ defined from skein theory can be view as a holomorphic section
$$Z_r'(S^3\backslash K)=\underset{1 \leqslant n \leqslant r}{\sum}\eta_{r}'\langle e_{n-1}\rangle_K \Phi_n\in\mathcal{H}^{alt}_{r+\frac{1}{2}}(j, \delta).$$
Recall that $|\langle e_{n-1}\rangle|=\left|[n]J_{K,n}\left(exp\left(\frac{2\pi i}{r+\frac{1}{2}}\right)\right)\right|$. Thus the $L^2$-norm of $Z'_r(S^3\backslash K)$ is 
\begin{equation}\label{l2}
  ||Z'_r(S^3\backslash K)||_{r+\frac{1}{2}}^2=\underset{1 \leqslant n \leqslant r}{\sum}|\eta_{r}'[n]|^2\left|J_{K,n}\left(exp\left(\frac{2\pi i}{r+\frac{1}{2}}\right)\right)\right|^2.  
\end{equation}

\section{Toward the volume conjecture}
Kashaev\cite{kashaev1997hyperbolic} defines a link invariant using quantum dilogarithm functions and conjectures the absolute value of this invariant of the link $L$ grow exponentially with rate equal to the simplicial volume of the complement of $L$. Murakami and Murakami\cite{murakami2001colored} show that Kashaev’s invariants coincide with the values of the colored Jones polynomials at a certain root of unity and reformulate this conjecture to
$$
\lim _{r \rightarrow+\infty} \frac{2 \pi}{r} \log \left|J_{L,r}\left(exp\left(\frac{2\pi i}{r}\right)\right)\right|\overset{?}{=}\operatorname{vol}\left(S^{3} \backslash L\right).
$$

From the study of asymptotic behavior of Turaev-Viro invariants, Tian Yang et.al.\cite{chen2018volume}\cite{detcherry2018turaev} make an argument that the evaluation in the above conjecture could be replace by $exp\left(\frac{2\pi i}{r+\frac{1}{2}}\right)$, that is 
$$
\lim _{r \rightarrow+\infty} \frac{2 \pi}{r} \log \left|J_{L,r}\left(exp\left(\frac{2\pi i}{r+\frac{1}{2}}\right)\right)\right|\overset{?}{=}\operatorname{vol}\left(S^{3} \backslash L\right).
$$

From many numerical results we observe that $$\underset{1 \leqslant n \leqslant r}{max}\left|J_{L,n}\left(exp\left(\frac{2\pi i}{r+\frac{1}{2}}\right)\right)\right|=\left|J_{L,r}\left(exp\left(\frac{2\pi i}{r+\frac{1}{2}}\right)\right)\right|,$$ and note that $\underset{1 \leqslant n \leqslant r}{max}|\eta_{r}'[n]|$ grows polynomially in $r$. Therefore, by (\hyperref[l2]{3.19
}) we state the following conjecture.

\begin{conj}
    For any knot $K$ in $S^3$, let $Z_r'(S^3\backslash K)\in\mathcal{H}^{alt}_{r+\frac{1}{2}}(j, \delta)$ be its $r$-th $SO(3)$-knot state. Then we have
    \begin{equation}
        \lim _{r \rightarrow+\infty} \frac{2 \pi}{r} \log ||Z'_r(S^3\backslash K)||_{r+\frac{1}{2}}=\operatorname{vol}\left(S^{3} \backslash K\right).
    \end{equation}
    
\end{conj}

Using Heegard decomposition Garoufalidis\cite{garoufalidis2011asymptotics} shows that for any 3-manifold $M$ the $SU(2)$-norm of $Z_r(M)$ grows only polynomial in $r$.  Furthermore, he shows that for any positive integer $l$, $\left|J_{L,r}\left(exp\left(\frac{2\pi i}{r+l}\right)\right)\right|$ grows
only polynomially in $r$. However, Garoufalidis's original argument is failed in the case of $SO(3)$ since the $SO(3)$-TQFT is not unitary, which is pointed out by ZhengHan Wang.

\section{Discussion} 
Why shall we study the volume conjecture geometrically? First, along the geometric path we can use some geometric analytic approaches, like the asymptotic analysis of Schwartz kernel, to study some quantum topological problems. Also, through some changes of basis we can represent the norms of knot states as other new geometric topological invariants.   

But the most important motivation is that in the geometric picture we can use the idea of holographic complexity, which play an important role in high energy and gravitational research recently, to the quantum space in TQFT. For example, using Nielsen’s framework\cite{dowling2008geometry} we can study some complexity metrics in the space of Berezin-Toeplitz operators. We may expect that these geometric quantities capture some topological information of hyperbolic 3-manifolds, inspired by the holographic complexity=volume conjecture.
Meanwhile, developing geometric approaches to the volume conjecture provides a mathematical path to this mysterious and charming program, as the CS/WZW correspondence is a mathematically rigorous holographic principal. 

Moreover, complexity emerges naturally through Mahler measure of A polynomials and topological complexity of hyperbolic 3-manifold. On the one hand, it has been shown for some examples that the Mahler measure of A-polynomial is equal to the simplicial volume of the corresponding link complement. Inspired by the AJ conjecture\cite{garoufalidis2004characteristic} we may expect that the Mahler measure is just the classical complexity. On the other hand, the simplicial volumes of 3-manifolds have complexity descriptions, like stable complexity and Matveev's topological complexity.
Also there are some arithmetic objects emerge in this picture, like L-functions and quantum modular forms\cite{zagier2010quantum}. We will discuss this program in more detail in the future.

\printbibliography

\end{document}